\documentclass[a4paper,11pt]{article}
\usepackage{amssymb,amsmath,amsthm}
\usepackage{mathtools}
\usepackage[UKenglish]{babel}
\usepackage{microtype}
\usepackage{a4wide}
\usepackage{thm-restate}
\usepackage{complexity}
\usepackage{comment}
\usepackage{marginnote}
\usepackage{hyperref} 
\hypersetup{colorlinks=true,citecolor=blue,linkcolor=blue,urlcolor=blue}
\usepackage{tikz}
\usepackage{bbm}

\theoremstyle{plain}
\newtheorem{theorem}{Theorem}
\newtheorem*{theorem*}{Theorem}
\newtheorem{lemma}[theorem]{Lemma}
\newtheorem*{lemma*}{Lemma}

\newtheorem*{proposition*}{Proposition}
\newtheorem{corollary}[theorem]{Corollary}
\newtheorem*{corollary*}{Corollary}

\theoremstyle{definition}

\newtheorem{remark}[theorem]{Remark}

\usepackage{algorithm}
\usepackage{algpseudocode}

\newcommand{\rfig}[1]{\mbox{Fig. \ref{fig:#1}}}
\newcommand{\ralg}[1]{\mbox{Algorithm~\ref{alg:#1}}}

\newcommand{\rlemma}[1]{\mbox{Lemma~\ref{lem:#1}}}
\newcommand{\rcorollary}[1]{\mbox{Corollary \ref{cor:#1}}}

\newcommand{\rtheorem}[1]{\mbox{Theorem \ref{thm:#1}}}
\newcommand{\requation}[1]{\mbox{Equation~(\ref{eq:#1})}}

\newcommand{\eps}{\varepsilon}
\newcommand{\Ex}[1]{\mathbb{E}\left[#1\right]}
\newcommand{\Var}[1]{\text{\normalfont $\mathbb{V}$ar}\left[#1\right]}
\renewcommand{\Pr}[1]{\mathbb{P}\left[#1\right]}
\newcommand{\size}[1]{\normalfont \textsf{size}(#1)}
\newcommand{\abs}[1]{\left|#1\right|} 
\newcommand{\MODRATIO}{\textsc{modular-ratio-max}\xspace}

\newcommand{\bA}{\ensuremath{\mathbf{A}}}
\newcommand{\bB}{\ensuremath{\mathbf{B}}}

\newcommand{\cB}{\ensuremath{\mathcal{B}}}
\newcommand{\cC}{\ensuremath{\mathcal{C}}}
\newcommand{\cD}{\ensuremath{\mathcal{D}}}
\newcommand{\cO}{\ensuremath{\mathcal{O}}}
\renewcommand{\cP}{\ensuremath{\mathcal{P}}}
\newcommand{\EX}{\ensuremath{\mathsf{EX}}}
\newcommand{\RR}{\ensuremath{\mathbb{R}}}
\newcommand{\Rnn}{\ensuremath{\mathbb{R}_{\ge 0}}}

\usepackage{tcolorbox}
\usepackage{enumitem}
\usepackage{xspace}

\begin{document}

\author{Jannik Kudla\\
University of Oxford\footnote{Now at Google.}\\
\texttt{janniku@google.com}
\and
Stanislav \v{Z}ivn\'y\\
University of Oxford\\
\texttt{standa.zivny@cs.ox.ac.uk}
}

\title{Sparsification of Monotone $k$-Submodular\\ Functions of Low
Curvature\thanks{This research was funded in whole by UKRI EP/X024431/1. For the
purpose of Open Access, the author has applied a CC BY public copyright licence
to any Author Accepted Manuscript version arising from this submission. All data
is provided in full in the results section of this paper.}}

\maketitle

\begin{abstract}
  Pioneered by Bencz\'ur and Karger for cuts in graphs~[STOC'96], sparsification
  is a fundamental topic with wide-ranging applications that has been studied,
  e.\,g., for graphs and hypergraphs, in a combinatorial and a spectral setting,
  and with additive and multiplicate error bounds. Rafiey and Yoshida recently
  considered sparsification of decomposable submodular functions~[AAAI'22]. We
  extend their work by presenting an efficient algorithm for a sparsifier for
  monotone $k$-submodular functions of low curvature. 
\end{abstract}

\section{Introduction}
\label{sec:intro}

The idea of ``sparsifying a graph'' (i.\,e., reducing the number of
edges) while preserving the value of all cuts goes back
to the influential paper~\cite{BK96}. 
The original motivation was to speed up algorithms for
cut problems and graph problems more generally. This concept turned out to be very
influential, with several generalisations and extensions
from graph cuts~\cite{Batson12:sicomp,BK15:sicomp,Andoni16:itcs}
to sketching~\cite{Ahn09:icalp,Andoni16:itcs}, sparsifiers for cuts in
hypergraphs~\cite{Kogan15:itcs,Newman13:sicomp}, spectral
sparsification~\cite{Spielman11:sicomp,Spielman04:stoc,Spielman11:sicomp-graph,Fung11:stoc,Lee18:sicomp,Soma19:soda,Kapralov21:focs}, sparsification of other predicates~\cite{Filtser17:sidma}, and additive
sparsification~\cite{Bansal19:focs}.

The cut function of a graph is an important example of a submodular function, which we define now.
Let $E$ be a finite set. A (set) function $F:2^E\to\RR$ defined on subsets of
$E$ is called \emph{submodular} if
\begin{equation}\label{ineq:submodular}
  F(S\cap T)+F(S\cup T)\ \leq\ F(S)+F(T)\qquad \forall S,T\subseteq E.
\end{equation}
Submodularity is a fundamental concept in combinatorial optimisation, with
applications across computer science and
economics~\cite{Nemhauser88:optimization,Topkis98:Supermodularity,schrijver2003combinatorial,Fujishige2005submodular}.
An equivalent definition of submodular functions captures the idea of
\emph{diminishing returns}.
\begin{equation}\label{ineq:submodular2}
  F(T\cup\{e\})-F(T)\ \leq\ F(S\cup\{e\})-F(S)\qquad \forall S\subseteq T\subseteq E, e\in E\setminus T.
\end{equation}

A set function $F$ is \emph{decomposable} if $F = \sum_{i
= 1}^N f_i$, where $f_i: 2^E \to \RR$ for each $i\in [N]=\{1,\ldots,N\}$ and
$E$ is a finite set of size $n=|E|$.
The cut function in a graph is an example of a decomposable submodular function,
in which the number $N$ of individual functions is equal to the number of edges
in the graph. (This is true even if the graph is directed and with nonnegative edge
weights.)

Rafiey and Yoshida~\cite{Rafiey22:aaai} considered the following natural
sparsification problem for $F$.\footnote{Each $f_i$ is represented by an oracle that returns, for any $S\subseteq E$, the value $f_i(S)$.}
Given tolerance parameters $\eps, \delta \in (0,
1)$, find a vector $w\in\RR^N$, called an $\eps$-\emph{sparsifier} (or just a
\emph{sparsifier}), such that the
function $F'=\sum_{i=1}^Nw_if_i$ satisfies, with probability at least
$1-\delta$,
\begin{equation}
  (1-\eps)F'(S)\ \leq\ F(S)\ \leq\ (1+\eps)F'(S)\qquad \forall S\subseteq E,
\end{equation}
and $\size{w}$, the set of nonzero entries of $w$, is as small as possible.
The idea in~\cite{Rafiey22:aaai} is, for each $i\in [N]$, to sample function
$f_i$ with probability $\kappa_i$ proportional to the ratio
\begin{equation}
  p_i = \max_{\substack{S\subseteq E\\F(S) \neq 0}} \frac{f_i(S)}{F(S)}.
\end{equation}
If $f_i$ is sampled, i.\,e., if it is decided that $f_i$ shall be part of the sparsifier, it is
included in the sparsifier with weight $1 / \kappa_i$, making its expected
weight equal to $\Ex{w_i} = \kappa \cdot 1 / \kappa_i = 1$ -- its weight in the
initial decomposition. In statistical terms, the sampling procedure is \emph{unbiased}.
The authors of~\cite{Rafiey22:aaai} showed the following.

\begin{theorem}[\cite{Rafiey22:aaai}]
    \label{thm:core-old}
    Let $F=\sum_{i=1}^N f_i$, where each $f_i:2^E\to\RR$ is submodular. For every $\eps,\delta\in(0,1)$ there is a vector $w \in \RR^N$ such that
    \begin{enumerate}[label = (\roman*)]
        \item $\Pr{\text{\normalfont $w$ is an $\eps$-sparsifier}} \ge 1 - \delta$;
        \item $\Ex{\size{w}} = \cO\left(\frac{n}{\eps^2} \sum_{i = 1}^N p_i \right)$,
        where $p_i = \max_{\substack{S \subseteq E\\F(S) \neq 0}} \frac{f_i(S)}{F(S)}$.
    \end{enumerate}
\end{theorem}

Computing and in many interesting cases even approximating the $p_i$'s is by far
the hardest step on the way to constructing a sparsifier. We shall refer to the
$p_i$'s as the \emph{peak contributions} -- since $p_i$ describes, on a scale
from $0$ to $1$, the maximum contribution of $f_i$ to $F$ when a set $S \subseteq E$ is chosen in favour of $f_i$.

Let $F=\sum_{i=1}^Nf_i$ be as in Theorem~\ref{thm:core-old}, i.\,e., with all
$f_i$'s submodular. 
Let $|\EX(\cB(f_i))|$ be the number of extreme points in the base
polyhedron of $f_i$~\cite{Fujishige2005submodular}, and let
$B=\max_{i\in[N]}|\EX(\cB(f_i))|$ (cf. Appendix~\ref{subsec:prelims} for precise
definitions).
The authors of~\cite{Rafiey22:aaai} claim that
\begin{equation}\label{ineq:peaks}
  \sum_{i=1}^Np_i\ \leq\ Bn,
\end{equation}
which implies by the virtue of Theorem~\ref{thm:core-old} the existence of a
sparsifier of expected size $\cO(\frac{Bn^2}{\eps^2})$.
As we will see later, this only holds if the $f_i$'s are monotone.
Using an $\cO(\sqrt{n})$-approximation of the peak contributions using the
ellipsoid method~\cite{Bai16:icml}, it is then established in \cite{Rafiey22:aaai}  that if
all $f_i$'s are not only submodular but also monotone, a sparsifier of expected size 
$\cO(\frac{Bn^{2.5}\log n}{\eps^2})$ can be found in randomised polynomial time,
assuming $(\ref{ineq:peaks})$ holds.
Here a function $F: 2^E\to\RR$ is called \emph{monotone} if $F(S) \leq F(T)$ for any $S\subseteq
T\subseteq E$.

\noindent\paragraph{Contributions}

As our main contribution, we provide a sparsification algorithm for decomposable monotone
$k$-submodular functions of low curvature.
As a starting point, we observe in Section~\ref{sec:core} (and prove in
Appendix~\ref{app:core-proofs}) that the sampling algorithm
from~\cite{Rafiey22:aaai} used to prove Theorem~\ref{thm:core-old} is largely independent
of submodularity, leading to a more general sparsification algorithm for
decomposable functions. 
Along the way, we establish a concentration bound
revealing that it is very unlikely that the resulting sparsifier exceeds
$(3/2)$-times the expected size.
In detail, consider a finite domain $\cD$, which is the power set $\cD = 2^E$ in the case of set
functions. Further suppose that $F: \cD \to \RR$ is decomposable as $F = \sum_{i =
1}^N f_i$, where $f_i: \cD \to \RR$ for each $i\in [N]$.\footnote{Each $f_i$ is
represented by an evaluation oracle that takes time $\cO(\text{EO}_i)$ to return
$f_i(S)$ for any $S\in\cD$.}

\begin{theorem}[Informal version of Theorem~\ref{thm:core}]
    \label{thm:core-informal}
    Let $F=\sum_{i=1}^N f_i$, where $f_i:\cD\to\RR$. For every $\eps,\delta\in(0,1)$ there is a vector $w \in \RR^N$ such that
    \begin{enumerate}[label = (\roman*)]
        \item 
          $\Pr{\text{\normalfont $w$ is an $\eps$-sparsifier}} \ge 1 - \delta$;
        \item 
          $\Ex{\size{w}} = \cO\left(\frac{\log{|\cD|} + \log{\frac{1}{\delta}}}{\eps^2} \sum_{i = 1}^N p_i \right)$,
        where $p_i = \max_{\substack{A \in \cD\\F(A) \neq 0}} \frac{f_i(A)}{F(A)}$;
        \item 
          $\Pr{\size{w} \le \frac{3}{2} \Ex{\size{w}}} \ge 1 - 4\eps^2$.
    \end{enumerate}
\end{theorem}

As our primary contribution, we use Theorem~\ref{thm:core-informal} to give in Section~\ref{sec:curvature} a sparsifier for decomposable monotone $k$-submodular functions of low curvature. 
As our secondary contribution, we clarify certain results on sparsification of submodular functions from~\cite{Rafiey22:aaai}. 
Firstly, we show 
that Inequality~(\ref{ineq:peaks}) claimed in~\cite{Rafiey22:aaai} is
incorrect by giving a counterexample. However, we show that Inequality~(\ref{ineq:peaks}) holds
under the additional assumption of monotonicity.  This is done in
Appendix~\ref{app:submodular}.
Secondly, in Appendix~\ref{app:bounded-arity}
we give a sparsifier for a class
of decomposable monotone submodular functions of bounded arity.

\medskip
For a natural number $k\geq 1$, a function $F:(k+1)^E\to\RR$ defined on $k$-tuples of
pairwise disjoints subsets of $E$ is called $k$-submodular if $f$ satisfies inequalities
similar to the submodularity inequality given in
Inequality~(\ref{ineq:submodular}). 
In detail, let $\bA=(A_1,\ldots,A_k)\in(k+1)^E$ be a $k$-tuple of pairwise disjoint
subsets of $E$, and similarly for $\bB=(B_1,\ldots,B_k)\in(k+1)^E$. Then,
$F:(k+1)^E\to\RR$ is called \emph{$k$-submodular} if
\begin{equation}
f(\bA\sqcap \bB)+f(\bA\sqcup \bB) \ \leq\ f(\bA)+f(\bB),
\end{equation}
where 
\begin{align}
\bA\sqcap \bB\ &=\ (A_1\cap B_1,\ldots, A_k\cap B_k), \\
\intertext{and}
\bA\sqcup \bB\ &=\ ((A_1\cup B_1)\setminus\smashoperator{\bigcup_{i\in\{2,\ldots,k\}}}\,(A_i\cup B_i),\ldots,(A_k\cup B_k)\setminus\smashoperator{\bigcup_{i\in\{1,\ldots,k-1\}}}\,(A_i\cup B_i)).
\end{align}
Under this definition, $1$-submodularity corresponds exactly to the standard
notion of submodularity for set functions as defined in
Inequality~(\ref{ineq:submodular}), and
similarly $2$-submodularity corresponds to \emph{bisubmodularity}~\cite{Bouchet87:greedy,Chandrasekaran88:pseudomatroids}.
The class of $k$-submodular functions was introduced in~\cite{Huber12:ksub} and played an
important role in the study of so-called finite-valued CSPs~\cite{HKP14:sicomp,KTZ15:sicomp}.
While minimising 1-submodular functions~\cite{Schrijver00:submodular,Iwata01:submodular}
and 2-submodular functions~\cite{Fujishige06:bisubmodular} given by evaluation
oracles can be done efficiently~\cite{Iwata08:sfm-survey}, the
complexity of the minimisation problem of $k$-submodular functions is open for
$k\geq 3$. On the other hand, the approximability of the
maximisation problem is well
understood for $k$-submodular functions~\cite{WZ16:talg,Iwata16:soda,Oshima21:sidma}, also for the monotone
case under cardinality constraint~\cite{Sakaue17:do}, in the streaming
model~\cite{Ene22:icml}, and other variants~\cite{Tang22:orl,Tang22:tcs,Pham22:jco}.

The definition of monotonicity for submodular functions gracefully extends
to $k$-submodular functions: $F:(k+1)^E\to\RR$ is called \emph{monotone} if
$F(\bA)\leq F(\bB)$ for all $\bA=(A_1,\ldots,A_k)\in(k+1)^E$ and
$\bB=(B_1,\ldots,B_k)\in(k+1)^E$ with $A_i\subseteq B_i$ for every $i\in[k]$.

An important concept studied in the context of submodular functions is that of
\emph{bounded curvature}~\cite{Conforti84:dam}.
For a monotone submodular function $F:2^E\to\Rnn$, the
\emph{curvature} (also called \emph{total curvature} in~\cite{Vondrak10}) $c_F$ of $F$
is defined by
\begin{equation}
  c_F\ =\ 1-\min_{S\subseteq E,e\in E\setminus S}\frac{\Delta_e
  F(S)}{\Delta_e F(\emptyset)},
\end{equation}
where $\Delta_e f(S)$ denotes the marginal gain of $e$ with respect to $S$,
i.\,e., 
\begin{equation}
  \Delta_e F(S)\ =\ F(S\cup\{e\})-F(S).
\end{equation}
In other words, the curvature compares the marginal gain of adding an element
of the ground set to any set and the empty set. Note that $c_F\in [0,1]$, with the
upper bound following from Inequality~(\ref{ineq:submodular2}). Also, $c_F=0$ holds
precisely when $f$ is modular, i.\,e., when Inequality~(\ref{ineq:submodular})
(equivalently, Inequality~(\ref{ineq:submodular2}))
holds with equality. We say that $f$ has \emph{low curvature} if $c_F<1$.
Intuitively, the curvature $c_F$ represents ``how
much the function curves''. The notion of curvature was extended from submodular
to $k$-submodular functions in~\cite{SantiagoS19:icml}, cf. also~\cite{ksub-curvature-def}.
In order to define it, we first need to introduce the notion of
marginal values for $k$-submodular functions, which is a natural generalisation
of the $k=1$ case. Let $F:(k+1)^E\to\RR$ be a $k$-submodular function. For 
$\bA=(A_1,\ldots,A_k)$, $i\in [k]$, and $e\in E\setminus\cup_{j\in [k]}A_j$, we
define the marginal gain of $e$ with respect to $\bA$ and $i$ as
\begin{equation}
  \Delta_{e,i} F(\bA)\ =\
  F(A_1,\ldots,A_{i-1},A_i\cup\{e\},A_{i+1},\ldots,A_k)-F(\bA).
\end{equation}
Then, the curvature $c_F$ of $F$ is defined as
\begin{equation}
  c_F\ =\ 1-\min_{i\in [k], e\in E, \bA\in(k+1)^{E\setminus\{e\}}}    
  \frac{\Delta_{e,i} F(\bA)}{\Delta_{e,i} F(\emptyset)}.
\end{equation}
As before, we say that $f$ has \emph{low curvature} if $c_F<1$.

As our main contribution, we will show that under the assumption of monotonicity
and low curvature one can efficiently approximate the peak contributions,
leading to an efficient execution of the sampling algorithm from
Section~\ref{sec:core}. Apart from being technically non-trivial, we also see
our work as a conceptual contribution to the area of sparsification by exploring
more general settings than previous works.

\section{The Core Algorithm}
\label{sec:core}

In this section, we will describe the core of all our sparsification algorithms
-- a randomised sampling routine initially described by Rafiey and
Yoshida~\cite{Rafiey22:aaai} for decomposable submodular functions. As alluded
to in Section~\ref{sec:intro}, we observe
that it is largely independent of submodularity, leading to a more general
sparsification algorithm for decomposable functions.  
Most of the presented material follows closely Section 3
in~\cite{Rafiey22:aaai} and details are deferred to Appendix~\ref{app:core-proofs}.

The algorithm we present here constructs an $\eps$-sparsifier for any
decomposable function $F = \sum_{i = 1}^N f_i: \cD \to \RR$
probabilistically.
As in~\cite{Rafiey22:aaai}, it relies on sampling functions with probabilities proportional to the ratios,
for each $i \in [N]$,
\begin{align}
    p_i = \max_{\substack{A \in \cD\\F(A) \neq 0}} \frac{f_i(A)}{F(A)}.
\end{align}

The procedure with all details is given in~\ralg{core}.

\begin{algorithm}[htbp]
    \begin{algorithmic}[1]
        \Require Function $F = f_1 + \dots + f_N$ with $f_i: \cD \to
        \RR$ given by evaluation oracles;
        error tolerance parameters $\eps, \delta \in (0, 1)$
        \Ensure Vector $w \in \RR^N$ such that
        \begin{itemize}
            \item $\Pr{\text{$w$ is an $\eps$-sparsifier}} \ge 1 - \delta$;
            \item $\Ex{\size{w}} = \cO\left(\frac{\log{|\cD|} + \log{\frac{1}{\delta}}}{\eps^2} \sum_{i = 1}^N p_i \right)$,
            where $p_i = \max_{\substack{A \in \cD\\F(A) \neq 0}} \frac{f_i(A)}{F(A)}$;
            \item $\Pr{\size{w} \le \frac{3}{2} \Ex{\size{w}}} \ge 1 - 4\eps^2$.
        \end{itemize}
        \State $w \gets (0, \dots, 0)$
        \State $\kappa \gets 3 \log{\left( \frac{2 |\cD|}{\delta} \right)} / \eps^2$
        \For{$i = 1, \dots, N$}
            \State \label{alg:core:pi} $p_i \gets \max_{\substack{A \in \cD\\F(A) \neq 0}} \frac{f_i(A)}{F(A)}$ \Comment{compute peak contribution (here: naively)}
            \State $\kappa_i \gets \min\{ 1, \kappa p_i \}$ \Comment{cap at $1$ as $\kappa_i$ is a probability}
            \State $w_i \gets \begin{cases}
                1/\kappa_i & \text{ with probability } \kappa_i \\
                0 & \text{ with probability } 1 - \kappa_i
            \end{cases}$ \Comment{sample weight of $f_i$}
        \EndFor
        \State \Return $w$
    \end{algorithmic}
    \caption{The Core Sparsification Algorithm}
    \label{alg:core}
\end{algorithm}

\begin{theorem}\label{thm:core}
    \ralg{core} outputs a vector $w \in \RR^N$ such that
    \begin{enumerate}[label = (\roman*)]
        \item
          $\Pr{\text{\normalfont $w$ is an $\eps$-sparsifier}} \ge 1 - \delta$;
        \item 
          $\Ex{\size{w}} = \cO\left(\frac{\log{|\cD|} + \log{\frac{1}{\delta}}}{\eps^2} \sum_{i = 1}^N p_i \right)$,
        where $p_i = \max_{\substack{A \in \cD\\F(A) \neq 0}} \frac{f_i(A)}{F(A)}$;
        \item \label{thm:core:size:variance0} 
          $\Pr{\size{w} \le \frac{3}{2} \Ex{\size{w}}} \ge 1 - 4\eps^2$.
    \end{enumerate}
\end{theorem}

\begin{remark}
    \ralg{core} can be invoked with $\delta = \cO(1 / n^c)$ so that it yields an $\eps$-sparsifier with high probability.
    This only influences the running time by a constant factor $c$ because of the dependence on $\log{\frac{1}{\delta}}$.
\end{remark}

\begin{remark}
    If the size of the sparsifier is of primary interest, running~\ralg{core} a
    couple of times and taking the smallest vector $w$ (with respect to $\size{w}$)
    leads to a procedure that, for any fixed $\eps > 0$, returns a sparsifier of size
    $\cO\left(\frac{\log{|\cD|} + \log{\frac{1}{\delta}}}{\eps^2} \sum_{i = 1}^N p_i \right)$
    after a logarithmic number of iterations. This is a consequence of Theorem~\ref{thm:core}\,\ref{thm:core:size:variance0}.
    Notice that it might be necessary to choose $\delta$ appropriately to also guarantee
    that the solution indeed is an $\eps$-sparsifier with high probability.
\end{remark}

\begin{corollary}
    \label{cor:core}
    In the setting of~\ralg{core}, let $\widehat{p}_1, \dots, \widehat{p}_N \in \Rnn$
    satisfy $\widehat{p}_i \ge p_i$ for all $i \in [N]$.
    If~\ralg{core} is executed with the $\widehat{p}_i$'s instead of $p_i = \max_{\substack{A \in \cD\\F(A)\neq 0}} \frac{f_i(A)}{F(A)}$
    in line~\ref{alg:core:pi}, it returns a vector $w \in \RR^N$ such that
    \begin{enumerate}[label = (\roman*)]
        \item $\Pr{\text{\normalfont $w$ is an $\eps$-sparsifier}} \ge 1 - \delta$;
        \item $\Ex{\size{w}} = \cO\left(\frac{\log{|\cD|} + \log{\frac{1}{\delta}}}{\eps^2} \sum_{i = 1}^N \widehat{p}_i \right)$;
        \item $\Pr{\size{w} \le \frac{3}{2} \Ex{\size{w}}} \ge 1 - 4\eps^2$.
    \end{enumerate}
\end{corollary}

Note that Corollary~\ref{cor:core} implies that any constant-factor approximation of the $p_i$'s will do the job,
leading to the same asymptotic bounds.
However, in the general setting of functions $f_1, \dots, f_N: \cD \to \RR$, there is no way of obtaining this much faster than the exact $p_i$'s.
Even in the case where all $f_i$'s are submodular, the $p_i$'s are hard to
approximate. 

\begin{remark}
    In general, the best upper bound we know on the peak contributions is $p_i \le 1$.
    Thus, \rcorollary{core} tells us that it is correct to invoke~\ralg{core} with $\widehat{p_i} = 1$ for all $i \in [N]$.
    Since $\kappa > 1$ for $\eps \in (0, 1)$, we then have $\kappa_i = \min\{ 1, \kappa p_i \} = 1$.
    This results in the initial decomposition, i.\,e., \ralg{core} essentially computes nothing -- a sparsifier is not for free!
\end{remark}

\begin{remark}
There are various ways to implement~\ralg{core}, leading to different running time bounds. The representation of the functions involved plays a key role here. In the most general scenario where no further assumptions are made, it is reasonable to assume each $f_i$ is represented by an evaluation oracle with response time $\cO(\mathsf{EO}_i)$. We may further assume an additional oracle for $F$ with response time $\cO(\mathsf{EO}_{\Sigma})$ -- it is in fact the case that, in many applications, $F(S)$ can be computed in a much faster way than by adding up all $f_i(S)$ for $i \in [N]$.

The main work to be done for~\ralg{core} to successfully construct a (small) sparsifier is the computation or approximation of the peak contributions. In the most general setting, we are required to compute
$p_i$ by iterating through all $A \in \cD$, which takes time at least $\Omega(|\cD|)$.
Hence, the running time of a naive implementation is in $\cO\left( N |\cD| \sum_{i = 1}^N \mathsf{EO}_i \right)$.
The main contribution of our work is to show how to approximate $p_i$'s
  efficiently for interesting cases, most notably for monotone $k$-submodular functions
  of low curvature.
\end{remark}

\section{Monotone \texorpdfstring{$k$}{k}-Submodular Functions of Low Curvature}
\label{sec:curvature}

Let $F: (k + 1)^E \to \Rnn$ be decomposable as $F=\sum_{i=1}^N f_i$ such that
each $f_i:(k+1)^E\to\Rnn$ is a non-negative monotone $k$-submodular function of low
curvature. It follows from the definitions that $F$ is also non-negative, monotone, $k$-submodular, and of
low curvature.

We will show how to approximate the peak contributions 
  $p_i = \max_{\bA \in (k + 1)^E} \frac{f_i(\bA)}{F(\bA)}$.

To this end, it suffices to approximate 
  $\max_{\bA \in (k + 1)^E} \frac{f(\bA)}{g(\bA)}$
for two monotone $k$-submodular functions $f, g: (k + 1)^E \to \Rnn$ of low curvature.
Given $\bA = (A_1, \dots, A_k) \in (k + 1)^E$, we define
\begin{equation}
  S_f(\bA) := \sum_{i = 1}^k \sum_{e \in A_i} \Delta_{e, i}\left( f
    \mid \varnothing, \dots, \varnothing \right) 
\end{equation}
and
\begin{equation}
    S_g(\bA) := \sum_{i = 1}^k \sum_{e \in A_i} \Delta_{e, i}\left( g \mid \varnothing, \dots, \varnothing \right).
\end{equation}
It turns out that $\frac{S_f(\bA) + f(\varnothing, \dots, \varnothing)}{S_g(\bA) + g(\varnothing, \dots, \varnothing)}$
approximates $\frac{f(\bA)}{g(\bA)}$ well.

\begin{lemma}
    \label{lem:approx}
    Let $\bA \in (k + 1)^E$ be an $(1 - \eps)$-approximate maximiser
    of $\frac{S_f(\bA) + f(\varnothing, \dots, \varnothing)}{S_g(\bA) + g(\varnothing, \dots, \varnothing)}$.
    Then
    \begin{align*}
        \frac{f(\bA)}{g(\bA)} \ge (1 - \eps)(1 - c_f)(1 - c_g)
        \frac{f(\bA^*)}{g(\bA^*)}
    \end{align*}
    for any $\bA^* \in (k + 1)^E$.
\end{lemma}

\noindent
Setting $\eps = 1/2$ in Lemma~\ref{lem:approx} gives a $\frac{1}{2} (1 - c_f) (1 - c_g)$-approximation, which is a constant factor if $c_f$ and $c_g$ are considered constants.
It remains to describe how
 $\max_{\bA \in (k + 1)^E} \frac{S_f(\bA) + f(\varnothing, \dots, \varnothing)}{S_g(\bA) + g(\varnothing, \dots, \varnothing)}$
can be approximated up to a factor of $1 - \eps$ (or $1/2$ specifically in our use case).
This is done by a reduction to the \MODRATIO problem,\footnote{Note that we allow $A$ and $B$ equal to zero, while the $x_i$'s and $y_i$'s are strictly positive. This is a somewhat technical requirement to avoid division by zero.}
defined below, for which
we design a fully polynomial approximation scheme (FPTAS).
\smallskip

\begin{tcolorbox}
    \underline{\MODRATIO} \linebreak
    \textbf{Given:} $x_1, \dots, x_n, y_1, \dots, y_n \in \RR_{> 0}$ and $A, B \in \Rnn$. \\
    \textbf{Want:} Index set $\varnothing \neq I \subseteq [n]$ such that $\varrho(I) := \frac{A + \sum_{i \in I} x_i}{B + \sum_{i \in I} y_i}$ is maximal.
\end{tcolorbox}

The reduction to \MODRATIO is now easy to describe. Recall that
\begin{equation}
    \frac{S_f(\bA) + f(\varnothing, \dots, \varnothing)}{S_g(\bA) + g(\varnothing, \dots, \varnothing)}
    = \frac{A + \sum_{i = 1}^k \sum_{e \in A_i} \Delta_{e, i}\left( f \mid \varnothing, \dots, \varnothing \right)}
    {B + \sum_{i = 1}^k \sum_{e \in A_i} \Delta_{e, i}\left( g \mid \varnothing, \dots, \varnothing \right)}
\end{equation}
with $A := f(\varnothing, \dots, \varnothing)$ and $B := g(\varnothing, \dots, \varnothing)$.
If we number the pairs $(e, i) \in E \times [k]$ in some arbitrary order $(e_1, i_1), \dots, (e_{nk}, i_{nk})$, we find that
maximising $\frac{S_f(\bA) + f(\varnothing, \dots, \varnothing)}{S_g(\bA) + g(\varnothing, \dots, \varnothing)}$
is the same as maximising
\begin{equation}
    \frac{A + \sum_{\ell \in I} \Delta_{e_{\ell}, i_{\ell}}\left( f \mid \varnothing, \dots, \varnothing \right)}
    {B + \sum_{\ell \in I} \Delta_{e_{\ell}, i_{\ell}}\left( g \mid \varnothing, \dots, \varnothing \right)}
\end{equation}
over all index sets $\varnothing \neq I \subseteq [nk]$ (the $I = \varnothing$ case can be checked manually or ignored
as it corresponds to $\bA = (\varnothing, \dots, \varnothing)$).
Since $f$ and $g$ are monotone, marginal gains are always non-negative, so
$\Delta_{e_{\ell}, i_{\ell}}\left( f \mid \varnothing, \dots, \varnothing \right) \ge 0$
and $\Delta_{e_{\ell}, i_{\ell}}\left( g \mid \varnothing, \dots, \varnothing \right) \ge 0$ for all $\ell \in [nk]$.
To satisfy the strict positivity as required in the definition of \MODRATIO, we can drop the marginal gains that are equal to $0$.
This will only shrink the problem.

Summing up, the algorithm to compute $p_i = \max_{\bA \in (k + 1)^E} \frac{f_i(\bA)}{F(\bA)}$ can be outlined as follows:
\begin{itemize}
  \item Compute a $(1/2)$-approximation $\bA^*$ to
    $\max_{\bA \in (k + 1)^E} \frac{S_f(\bA) + f(\varnothing, \dots, \varnothing)}{S_g(\bA) + g(\varnothing, \dots, \varnothing)}$
    via the \MODRATIO reduction and Algorithm \ref{alg:fptas:mod:ratio}.
    \item Let $\widehat{p}_i := \frac{2}{(1 - c_f)(1 - c_g)}
      \frac{f(\bA^*)}{g(\bA^*)}$.
\end{itemize}
It is guaranteed that $\widehat{p}_i \ge p_i$ by Lemma \ref{lem:approx}. Moreover, $\widehat{p}_i = \cO(1) \cdot p_i$, so
we get a sparsifier of an expected size that matches the existence result by applying the core algorithm.
Moreover, the algorithm runs in polynomial time, as we will see in Lemma \ref{lem:iteration:bound}.

\subsection{FPTAS}

If $A = B = 0$, \MODRATIO has a very simple solution: We can just take $I = \{ i \}$ for an index $i$ that maximises $x_i / y_i$.
However, this is not optimal in general as the following example shows. Let $A = 1$, $B = 100$ and $x_1 = 2$, $y_1 = 3$, $x_2 = 1$, $y_2 = 1$.
Clearly, the ratio $x_i / y_i$ is maximised for $i = 2$, leading to an overall $\varrho$-value of
\begin{align*}
    \varrho(\{ 2 \}) = \frac{1 + 1}{100 + 1} = \frac{2}{101}
\end{align*}
as opposed to
\begin{align*}
    \varrho(\{ 1, 2 \}) = \frac{1 + 2 + 1}{100 + 3 + 1} = \frac{4}{104},
\end{align*}
which is clearly larger. In fact, it is not hard to see that the maximiser of $x_i / y_i$ does not even provide a constant-factor approximation;
it may end up arbitrarily bad compared to an optimal solution. This indicates that we need to do something else.

The solution is an FPTAS based on binary search, outlined in Algorithm \ref{alg:fptas:mod:ratio}.
This is possible because we can easily solve the associated decision problem:
Given a target value $\lambda \in \RR$, does there exist an index set $\varnothing \neq I \subseteq [n]$ such that $\varrho(I) \ge \lambda$?

The decision problem is simplified by algebraic equivalence transformations:
\begin{align*}
    \varrho(I) \ge \lambda
    \iff \frac{A + \sum_{i \in I} x_i}{B + \sum_{i \in I} y_i} \ge \lambda
    \iff A - B\lambda + \sum_{i \in I} \left( x_i - \lambda y_i \right) \ge 0
\end{align*}
To see if the last expression is non-negative for any non-empty index set $I$,
we consider the indices in non-increasing order of the quantities $x_i - \lambda y_i$.
We have to take first index in this order (as $I \neq \varnothing$ is required) and will then take all remaining indices $i$ for which $x_i - \lambda y_i$
is positive. This maximises the LHS over all $I \neq \varnothing$. If this maximum is non-negative, we know that $\varrho(I) \ge \lambda$ by the
above equivalences.

Let $m := \min\{ x_1, \dots, x_n, y_1, \dots, y_n \}$ and $M := \max\{ x_1, \dots, x_n, y_1, \dots, y_n \}$.
For any non-empty index set $I$, we always have
\begin{align}
    \label{eq:rho:lower}
    \varrho(I) = \frac{A + \sum_{i \in I} x_i}{B + \sum_{i \in I} y_i} \ge \frac{A + m}{B + n M}
\end{align}
as well as
\begin{align}
    \label{eq:rho:upper}
    \varrho(I) = \frac{A + \sum_{i \in I} x_i}{B + \sum_{i \in I} y_i} \le \frac{A + n M}{B + m},
\end{align}
so we can initialise the binary search with $\varrho^- = \frac{A + m}{B + n M}$ and $\varrho^+ = \frac{A + n M}{B + m}$.
Once the interval $\left[ \varrho^{-}, \varrho^{+} \right]$ has length at most $\eps \frac{A + m}{B + n M}$, we know that the multiplicative
error is at most $\eps$. Since the interval size halves in each step, this point is reached after no more than
$\left\lceil \log{\frac{1}{\eps}} + 2 \left( \log{n} + \log{\frac{M}{m}} \right) \right\rceil$ iterations, as the following lemma shows.

\begin{lemma}
    \label{lem:iteration:bound}
    The binary search in Algorithm \ref{alg:fptas:mod:ratio} terminates
    in $k := \left\lceil \log{\frac{1}{\eps}} + 2 \left( \log{n} + \log{\frac{M}{m}} \right) \right\rceil$ iterations.
    Moreover, the final set $I$ satisfies $\varrho(I) \ge (1 - \eps) \varrho(I^*)$ for any $\varnothing \neq I^* \subseteq [n]$.
\end{lemma}
\begin{proof}
    To see the iteration bound, we note that $\abs{\varrho^+ - \varrho^-}$ shrinks by a factor of $2$ in each iteration.
    Thus, after $k$ iterations, it holds
    \begin{align*}
        \abs{\varrho^+ - \varrho^-} \le \frac{1}{2^k} \abs{\frac{A + n M}{B + m} - \frac{A + m}{B + n M}} \le \frac{1}{2^k} \frac{A + n M}{B + m}.
    \end{align*}
    We want this to be $\le \eps \frac{A + m}{B + n M}$, which is equivalent to
    \begin{align}
        \label{eq:k:large:enough}
        \frac{1}{2^k} \frac{A + n M}{B + m} \le \eps \frac{A + m}{B + n M}
        \iff 2^k \ge \frac{1}{\eps} \frac{A + n M}{A + m} \frac{B + n M}{B + m}.
    \end{align}
    Since $n M \ge m$, we have $\frac{A + n M}{A + m} \le \frac{n M}{m}$ as well as $\frac{B + n M}{B + m} \le \frac{n M}{m}$, hence
    \begin{align*}
        \frac{1}{\eps} \frac{A + n M}{A + m} \frac{B + n M}{B + m}
        \le \frac{1}{\eps} \left( \frac{n M}{m} \right)^2,
    \end{align*}
    so it suffices to satisfy $2^k \ge \frac{1}{\eps} \left( \frac{n M}{m}
    \right)^2$ in order for Equation~(\ref{eq:k:large:enough}) to hold.
    This is indeed satisfied for any $k \ge \log{\frac{1}{\eps}} + 2 \left( \log{n} + \log{\frac{M}{m}} \right)$, showing the iteration bound.

    For the error bound, let $\varnothing \neq I^* \subseteq [n]$ be arbitrary.
    By Equation~(\ref{eq:rho:lower}), we know that $\varrho(I^*) \ge \frac{A + m}{B + n M}$.
    Moreover, the binary search preserves two invariants:
    \begin{enumerate}[label = (\roman*)]
        \item \label{inv:1} $\varrho^- \le \varrho(I) \le \varrho^+$ for the set $I$ in Algorithm \ref{alg:fptas:mod:ratio}, and
        \item \label{inv:2} There is no set $\varnothing \neq I' \subseteq [n]$ with $\varrho(I') > \varrho^+$.
    \end{enumerate}
    Combining both and the fact that $\abs{\varrho^+ - \varrho^-} \le \eps \frac{A + m}{B + n M}$ at termination, we conclude
    \begin{align*}
        \varrho(I) \ge \varrho^- = \varrho^+ - \left( \varrho^+ - \varrho^- \right)
        \ge \varrho^+ - \eps \frac{A + m}{B + n M}
        \ge \varrho(I^*) - \eps \varrho(I^*) = (1 - \eps) \varrho(I^*),
    \end{align*}
    where we also exploited $\varrho(I') \ge \frac{A + m}{B + n M}$
    (Equation~(\ref{eq:rho:lower})) and we used invariant \ref{inv:2}
    as $\varrho(I^*) \le \varrho^+$.
\end{proof}

\begin{algorithm}[htbp]
    \begin{algorithmic}[1]
        \Require \MODRATIO instance $x_1, \dots, x_n, y_1, \dots, y_n \in \RR_{> 0}$ and $A, B \in \Rnn$;
        error tolerance $\eps > 0$.
        \Ensure Index set $\varnothing \neq I \subseteq [n]$ such that $\varrho(I) \ge (1 - \eps) \varrho(I^*)$ for all $\varnothing \neq I^* \subseteq [n]$.
        \Procedure{check}{$\lambda$}
            \State Sort indices such that $x_{i_1} - \lambda y_{i_1} \ge x_{i_2} - \lambda y_{i_2} \ge \dots \ge x_{i_n} - \lambda y_{i_n}$
            \State $I \gets \{ i_1 \}$
            \State $S \gets \left( A - \lambda B \right) + \left( x_{i_1} - \lambda y_{i_1} \right)$
            \For{$\ell = 2$ \textbf{to} $n$}
                \If{$x_{i_\ell} - \lambda y_{i_\ell} > 0$}
                    \State $I \gets I \cup \{ i_{\ell} \}$
                    \State $S \gets S + \left( x_{i_{\ell}} - \lambda y_{i_{\ell}} \right)$
                \EndIf
            \EndFor
            \State \Return $\displaystyle \begin{cases}
                I & \text{ if } S \ge 0 \\
                \bot & \text{ otherwise }
            \end{cases}$
        \EndProcedure
        \State $I \gets \{ 1 \}$ \Comment{initialise with arbitrary feasible solution}
        \State $m \gets \min\left\{ x_1, \dots, x_n, y_1, \dots, y_n \right\}$
        \State $M \gets \max\left\{ x_1, \dots, x_n, y_1, \dots, y_n \right\}$
        \State $\varrho^{-} \gets \frac{A + m}{B + n M}$
        \State $\varrho^{+} \gets \frac{A + n M}{B + m}$
        \While{$\abs{\varrho^{+} - \varrho^{-}} > \eps \frac{A + m}{B + n M}$} \Comment{binary search till multiplicative error is $\le \eps$}
            \State $\lambda \gets \frac{1}{2} \left( \varrho^{-} + \varrho^{+} \right)$
            \State $I_{\lambda} \gets \Call{check}{\lambda}$
            \If{$I_{\lambda} = \bot$}
                \State $\varrho^{+} \gets \lambda$
            \Else
                \State $I \gets I_{\lambda}$
                \State $\varrho^{-} \gets \lambda$
            \EndIf
        \EndWhile
        \State \Return $I$
    \end{algorithmic}
    \caption{FPTAS for \MODRATIO}
    \label{alg:fptas:mod:ratio}
\end{algorithm}

We remark that, if all input numbers are integers, we get an iteration bound of
\begin{align*}
    k \le \left\lceil \log{\frac{1}{\eps}} + 2 \left( \log{n} + \log{U} \right) \right\rceil
    = \cO\left( \log{\frac{1}{\eps}} + \log{n} + \log{U} \right),
\end{align*}
where $U$ is the largest number that occurs as part of the input.
We have now solved \MODRATIO.

\subsection{Proof of Lemma~\ref{lem:approx}}

In the rest of his section, we will prove Lemma~\ref{lem:approx}. But first we
need a helpful lemma.

\begin{lemma}
    \label{lem:mod:ratio:approx}
    Let $f, g: (k + 1)^E \to \Rnn$ 
    be monotone $k$-submodular of low curvature.
    Then
    \begin{align*}
        \frac{(1 - c_f) S_f(\bA) + f(\varnothing, \dots, \varnothing)}{S_g(\bA) + g(\varnothing, \dots, \varnothing)}
        \le \frac{f(\bA)}{g(\bA)} \le
        \frac{S_f(\bA) + f(\varnothing, \dots, \varnothing)}{(1 - c_g) S_g(\bA) + g(\varnothing, \dots, \varnothing)}
    \end{align*}
    for all $\bA \in (k + 1)^E$.
\end{lemma}
\begin{proof}
    Fix $\bA \in (k + 1)^E$
    and label the elements of $E$ in such a way that $A_i = \left\{ e^{(i)}_1, \dots, e^{(i)}_{a_i} \right\}$ for each $1 \le i \le k$, where $a_i = |A_i|$.
    Now,
    \begin{align*}
        f(\bA) - f(\varnothing, \dots, \varnothing)
        = \sum_{i = 1}^k \sum_{j = 1}^{a_i} \Delta_{e^{(i)}_j, i}\left( f \mid A_1, \dots, A_{i - 1}, \{ e^{(i)}_1, \dots, e^{(i)}_{j - 1} \},
        \varnothing, \dots, \varnothing \right).
    \end{align*}
    The $\Delta_{e^{(i)}_j, i}\left( f \mid A_1, \dots, A_{i - 1}, \{ e^{(i)}_1, \dots, e^{(i)}_{j - 1} \},
    \varnothing, \dots, \varnothing \right)$ terms can be estimated in both directions. By the diminishing returns property, we know that
    \begin{align*}
        \Delta_{e^{(i)}_j, i}\left( f \mid A_1, \dots, A_{i - 1}, \{ e^{(i)}_1, \dots, e^{(i)}_{j - 1} \},
        \varnothing, \dots, \varnothing \right)
        \le \Delta_{e^{(i)}_j, i}\left( f \mid \varnothing, \dots, \varnothing \right),
    \end{align*}
    while we conclude
    \begin{align*}
        \Delta_{e^{(i)}_j, i}\left( f \mid A_1, \dots, A_{i - 1}, \{ e^{(i)}_1, \dots, e^{(i)}_{j - 1} \},
        \varnothing, \dots, \varnothing \right)
        \ge (1 - c_f) \Delta_{e^{(i)}_j, i}\left( f \mid  \varnothing, \dots, \varnothing \right)
    \end{align*}
    from the curvature of $f$. Since
    \begin{align*}
        S_f(\bA) := \sum_{i = 1}^k \sum_{j = 1}^{a_i} \Delta_{e^{(i)}_j, i}\left( f \mid  \varnothing, \dots, \varnothing \right),
    \end{align*}
    we see that $(1 - c_f) S_f(\bA) \le f(\bA) - f(\varnothing, \dots, \varnothing) \le S_f(\bA)$ after combining both inequalities.
    Analogously, we derive the $(1 - c_g) S_g(\bA) \le f(\bA) - f(\varnothing, \dots, \varnothing) \le S_g(\bA)$ for $g$. Next,
    \begin{align*}
        \frac{f(\bA)}{g(\bA)}
        = \frac
            {f(\bA) - f(\varnothing, \dots, \varnothing) + f(\varnothing, \dots, \varnothing)}
            {g(\bA) - g(\varnothing, \dots, \varnothing) + g(\varnothing, \dots, \varnothing)}
        \le \frac{S_f(\bA) + f(\varnothing, \dots, \varnothing)}{(1 - c_g) S_g(\bA) + g(\varnothing, \dots, \varnothing)}
    \end{align*}
    and
    \begin{align*}
            \frac{f(\bA)}{g(\bA)}
            = \frac
                {f(\bA) - f(\varnothing, \dots, \varnothing) + f(\varnothing, \dots, \varnothing)}
                {g(\bA) - g(\varnothing, \dots, \varnothing) + g(\varnothing, \dots, \varnothing)}
            \ge \frac{(1 - c_f) S_f(\bA) + f(\varnothing, \dots, \varnothing)}{S_g(\bA) + g(\varnothing, \dots, \varnothing)},
    \end{align*}
    squeezing $\frac{f(\bA)}{g(\bA)}$ between
    $\frac{(1 - c_f) S_f(\bA) + f(\varnothing, \dots, \varnothing)}{S_g(\bA) + g(\varnothing, \dots, \varnothing)}$ and
    $\frac{S_f(\bA) + f(\varnothing, \dots, \varnothing)}{(1 - c_g) S_g(\bA) + g(\varnothing, \dots, \varnothing)}$.
\end{proof}

\begin{lemma*}[Lemma~\ref{lem:approx} restated]
    Let $\bA \in (k + 1)^E$ be an $(1 - \eps)$-approximate maximiser
    of $\frac{S_f(\bA) + f(\varnothing, \dots, \varnothing)}{S_g(\bA) + g(\varnothing, \dots, \varnothing)}$.
    Then
    \begin{align*}
        \frac{f(\bA)}{g(\bA)} \ge (1 - \eps)(1 - c_f)(1 - c_g) \frac{f(\bA^*)}{g(\bA^*)}
    \end{align*}
    for any $\bA^* \in (k + 1)^E$.
\end{lemma*}
\begin{proof}
    Note that $\frac{f(\bA)}{g(\bA)}$ and $\frac{f(\bA^*)}{g(\bA^*)}$ are included in the ranges stated by
    Lemma \ref{lem:mod:ratio:approx}, i.\,e.,
    \begin{align*}
        \frac{(1 - c_f) S_f(\bA^*) + f(\varnothing, \dots, \varnothing)}{S_g(\bA^*) + g(\varnothing, \dots, \varnothing)}
        \le \frac{f(\bA^*)}{g(\bA^*)} \le
        \frac{S_f(\bA^*) + f(\varnothing, \dots, \varnothing)}{(1 - c_g) S_g(\bA^*) + g(\varnothing, \dots, \varnothing)}
    \end{align*}
    and
    \begin{align*}
        \frac{(1 - c_f) S_f(\bA) + f(\varnothing, \dots, \varnothing)}{S_g(\bA) + g(\varnothing, \dots, \varnothing)}
        \le \frac{f(\bA)}{g(\bA)} \le
        \frac{S_f(\bA) + f(\varnothing, \dots, \varnothing)}{(1 - c_g) S_g(\bA) + g(\varnothing, \dots, \varnothing)}.
    \end{align*}
    Combining this with the fact that $\bA$ is an $(1 - \eps)$-approximate maximiser and the non-negativity of $f$ and $g$, we conclude that
    \begin{align*}
        \frac{f(\bA)}{g(\bA)}
        &\ge \frac{(1 - c_f) S_f(\bA) + f(\varnothing, \dots, \varnothing)}{S_g(\bA) + g(\varnothing, \dots, \varnothing)} \\
        &\ge (1 - c_f) \frac{S_f(\bA) + f(\varnothing, \dots, \varnothing)}{S_g(\bA) + g(\varnothing, \dots, \varnothing)} \\
        &\ge (1 - c_f)(1 - \eps) \frac{S_f(\bA^*) + f(\varnothing, \dots, \varnothing)}{S_g(\bA^*) + g(\varnothing, \dots, \varnothing)} \\
        &\ge (1 - \eps)(1 - c_f)(1 - c_g) \frac{S_f(\bA^*) + f(\varnothing, \dots, \varnothing)}{(1 - c_g) S_g(\bA^*) + g(\varnothing, \dots, \varnothing)} \\
        &\ge (1 - \eps)(1 - c_f)(1 - c_g) \frac{f(\bA^*)}{g(\bA^*)}.
    \end{align*}
\end{proof}

{\small
\bibliographystyle{plainurl}
\bibliography{kz}
}

\appendix

\section{Proofs of Theorem~\ref{thm:core} and Corollary~\ref{cor:core}}
\label{app:core-proofs}

\begin{theorem*}[Theorem~\ref{thm:core} restated]
    \ralg{core} outputs a vector $w \in \RR^N$ such that
    \begin{enumerate}[label = (\roman*)]
        \item \label{thm:core:sparsifier} 
          $\Pr{\text{\normalfont $w$ is an $\eps$-sparsifier}} \ge 1 - \delta$;
        \item \label{thm:core:size} 
          $\Ex{\size{w}} = \cO\left(\frac{\log{|\cD|} + \log{\frac{1}{\delta}}}{\eps^2} \sum_{i = 1}^N p_i \right)$,
        where $p_i = \max_{\substack{A \in \cD\\F(A) \neq 0}} \frac{f_i(A)}{F(A)}$;
        \item \label{thm:core:size:variance} 
          $\Pr{\size{w} \le \frac{3}{2} \Ex{\size{w}}} \ge 1 - 4\eps^2$.
    \end{enumerate}
\end{theorem*}

\begin{proof}
    We show part~\ref{thm:core:sparsifier} first by showing that each $A \in \cD$ satisfies
    \begin{align}
        \label{eq:eps:bound:single:set}
        \Pr{(1 - \eps) F(A) \le F'(A) \le (1 + \eps) F(A)} \ge 1 - \frac{\delta}{|\cD|}.
    \end{align}
    An application of the union bound over all $|\cD|$ elements\footnote{Over all elements $A \in \cD$ such that $F(A) \neq 0$
    to be precise. If $F(A) = 0$, we know that $f_i(A) = 0$ for all $1 \le i \le N$, so \requation{same:expectation}
    and \requation{chernoff:goal} trivially hold. We do not consider such sets $A$ in our analysis and encourage the reader to think of $A \in \cD$
    with $F(A) > 0$ for the rest of the analysis.} of the domain then proves the claim.
    Note that
    \begin{align}
        \label{eq:same:expectation}
        \Ex{F'(A)} = \Ex{\sum_{i = 1}^{N} w_i f_i(A)}
        = \sum_{i = 1}^{N} \Ex{w_i} f_i(A)
        = \sum_{i = 1}^{N} f_i(A)
        = F(A)
    \end{align}
    since $\Ex{w_i} = \kappa_i \cdot 1 / \kappa_i + (1 - \kappa_i) \cdot 0 = 1$.
    Therefore, the inequality indicating ``success'' above is equivalent to
    $(1 - \eps) \Ex{F'(A)} \le F'(A) \le (1 + \eps) \Ex{F'(A)}$, i.\,e.,
    that $F'(A)$ lies within a factor of $1 \pm \eps$ of its expectation.
    The next step is to establish
    \begin{align}
        \label{eq:chernoff:goal}
        \Pr{\abs{F'(A) - \Ex{F'(A)}} \ge \eps \cdot \Ex{F'(A)}} \le 2 e^{-\eps^2 \kappa / 3}.
    \end{align}
    This is essentially done by recognising that the terms $w_i f_i(A)$ can be regarded as independent random variables with a bounded range $[0, a]$
    for some $a \le F(A) / \kappa$.
    However, we cannot simply mimic the proof by Rafiey and Yoshida~\cite{Rafiey22:aaai} here as it includes an invalid step.
    They claim that $\max_{1 \le i \le N} w_i f_i(A) = \max_{1 \le i \le N} \frac{f_i(A)}{\kappa p_i}$,
    which is not necessarily true when $\kappa_i = 1$ (remember that $\kappa_i = \min\{ 1, \kappa p_i \}$).
    In this case, $w_i = 1/\kappa_i = 1$, so $w_i f_i(A) = f_i(A)$, which exceeds $\frac{f_i(A)}{\kappa p_i}$ for $\kappa p_i > 1$.
    This in fact shows that their inequality goes in the wrong direction.

    Fortunately, this issue can be circumvented by considering the indices $i$ with $\kappa p_i > 1$ separately.
    Let $I := \{ i \in \{ 1, \dots, N \} \mid \kappa p_i > 1 \}$ and $\bar{I} := \{ 1, \dots, N \} \setminus I$.
    Moreover, define $Z_i := w_i f_i(A)$ for all $1 \le i \le N$, $Z := \sum_{i \in I} Z_i$ and $\bar{Z} := \sum_{i \in \bar{I}} Z_i$. We have
    \begin{align*}
        F'(A) = \sum_{i = 1}^N Z_i = \sum_{i \in I} Z_i + \sum_{i \in \bar{I}} Z_i = Z + \bar{Z}
    \end{align*}
    where the first part $\sum_{i \in I} Z_i = \sum_{i \in I} f_i(A)$
    is deterministic as we have $w_i = 1$ with probability $\kappa_i = \min\{ 1, \kappa p_i \} = 1$ for each $i \in I$.
    Thus,
    \begin{align*}
        F'(A) - \Ex{F'(A)}
        &= \sum_{i = 1}^N w_i f_i(A) - \sum_{i = 1}^N f_i(A) \\
        &= \sum_{i \in I} w_i f_i(A) + \sum_{i \in \bar{I}} w_i f_i(A) - \sum_{i \in I} f_i(A) - \sum_{i \in \bar{I}} f_i(A) \\
        &= Z + \sum_{i \in \bar{I}} w_i f_i(A) - Z - \sum_{i \in \bar{I}} f_i(A) \\
        &= \sum_{i \in \bar{I}} w_i f_i(A) - \sum_{i \in \bar{I}} f_i(A) \\
        &= \sum_{i \in \bar{I}} Z_i - \sum_{i \in \bar{I}} \Ex{Z_i} \\
        &= \bar{Z} - \Ex{\bar{Z}}
    \end{align*}
    by \requation{same:expectation}, the partitioning $\{ 1, \dots, N \} = I \cup \bar{I}$,
    definition of the $Z_i$'s, $Z$ and $\bar{Z}$, and linearity of expectation.
    We conclude that \requation{chernoff:goal} is equivalent to
    \begin{align}
        \label{eq:chernoff:goal:2}
        \Pr{\abs{\bar{Z} - \Ex{\bar{Z}}} \ge \eps \cdot \Ex{F'(A)}} \le 2 e^{-\eps^2 \kappa / 3}.
    \end{align}
    Since $f_i(A) \ge 0$ for all $1 \le i \le N$, observe that
    \begin{align}
        \Ex{\bar{Z}} = \sum_{i \in \bar{I}} f_i(A)
        \le \sum_{i \in \bar{I}} f_i(A) + \sum_{i \in I} f_i(A)
        = \sum_{i = 1}^N f_i(A) = F(A) = \Ex{F'(A)},
    \end{align}
    where the last step is by \requation{same:expectation}. We now bound $\Pr{\abs{\bar{Z} - \Ex{\bar{Z}}} \ge \eps \cdot \Ex{F'(A)}}$
    by an application of the Chernoff-Hoeffding bound~\cite{randomized-algorithms} in the version of Theorem 2.2 in~\cite{Rafiey22:aaai}.
    A check of the preconditions reveals that the $Z_i$'s are independent random variables (since the $w_i$'s are sampled independently).
    Letting $a := \max_{i \in \bar{I}} Z_i$, the $Z_i$'s with $i \in \bar{I}$ all have range $[0, a]$.
    Choosing $\mu := \Ex{F'(A)} \ge \Ex{\bar{Z}}$, we conclude
    \begin{align}
        \label{eq:chernoff:goal:2:almost}
        \begin{split}
        &\Pr{\abs{\bar{Z} - \Ex{\bar{Z}}} \ge \eps \mu} \le 2 e^{-\frac{\eps^2 \mu}{3a}} \\
        \iff &\Pr{\abs{\bar{Z} - \Ex{\bar{Z}}} \ge \eps \cdot \Ex{F'(A)}} \le 2 e^{-\frac{\eps^2 F(A)}{3a}}.
        \end{split}
    \end{align}
    This is very close to \requation{chernoff:goal:2}. Indeed, the last we need is an upper bound on $a$.
    To accomplish this, note that
    \begin{align*}
        Z_i = w_i f_i(A) = \frac{f_i(A)}{\kappa p_i} \le \frac{f_i(A)}{\kappa \max_{S \subseteq E} \frac{f_i(S)}{F(S)}}
        \le \frac{f_i(A)}{\kappa \frac{f_i(A)}{F(A)}} = \frac{F(A)}{\kappa}
    \end{align*}
    for all $i \in \bar{I}$. Thus, maximising over all $i \in \bar{I}$ reveals $a = \max_{i \in \bar{I}} Z_i \le F(A) / \kappa$.
    Substituting this into \requation{chernoff:goal:2:almost} yields \requation{chernoff:goal:2}.
    The RHS $2 e^{-\eps^2 \kappa / 3}$ becomes $\le \frac{\delta}{|\cD|}$ when choosing $\kappa = 3 \log{\left( \frac{2|\cD|}{\delta} \right)} / \eps^2$
    as in~\ralg{core} above, proving \requation{eps:bound:single:set}.
    
    Part~\ref{thm:core:size} is a simple computation. The events $\{ w_i \neq 0 \}$ occur with probability $\kappa_i$, so
    \begin{align*}
        \Ex{\size{w}} = \Ex{\sum_{i = 1}^{N} \left[ w_i \neq 0 \right]}
        = \sum_{i = 1}^{N} \Pr{w_i \neq 0}
        = \sum_{i = 1}^{N} \kappa_i
        \le \sum_{i = 1}^{N} \kappa p_i
        = \kappa \sum_{i = 1}^{N} p_i,
    \end{align*}
    where $\kappa = 3 \log{\left( \frac{2 |\cD|}{\delta} \right)} / \eps^2
    = \cO\left( \frac{\log{|\cD|} + \log{\frac{1}{\delta}}}{\eps^2} \right)$
    by the choice in~\ralg{core}.
    Hence the overall bound $\Ex{\size{w}} = \cO\left( \frac{\log{|\cD|} + \log{\frac{1}{\delta}}}{\eps^2} \sum_{i = 1}^{N} p_i \right)$.

    Part~\ref{thm:core:size:variance} is shown as follows.
    Let $[w_i \neq 0]$ be an indicator random variable for the event $\{ w_i \neq 0 \}$, i.\,e., $[w_i \neq 0] = 1$ if $w_i \neq 0$
    and $[w_i \neq 0] = 0$ if $w_i = 0$.
    The variance of a sum of independent random variables is the sum of their variances.
    Thus,
    \begin{align*}
        \Var{\size{w}} = \Var{\sum_{i = 1}^N [w_i \neq 0]} = \sum_{i = 1}^N \Var{[w_i \neq 0]}
    \end{align*}
    where $\Var{[w_i \neq 0]} = \Ex{[w_i \neq 0]^2} - \Ex{[w_i \neq 0]}^2 = \kappa_i - \kappa_i^2$ for each $1 \le i \le N$.
    It follows that $\Var{\size{w}} = \sum_{i = 1}^N \left( \kappa_i - \kappa_i^2 \right) \le \sum_{i = 1}^N \kappa_i = \Ex{\size{w}}$.

    Thus, by Chebyshev's inequality,
    \begin{align*}
        \Pr{\abs{\size{w} - \Ex{\size{w}}} \ge \frac{1}{2} \Ex{\size{w}}} \le \frac{\Var{\size{w}}}{\left( \frac{1}{2} \Ex{\size{w}} \right)^2}
        \le \frac{4}{\Ex{\size{w}}},
    \end{align*}
    which is $\le 4\eps^2$ because $\Ex{\size{w}} \ge 1/\eps^2$:
    If $\kappa_i = 1$ for some $i$ (i.\,e. $I \neq \varnothing$), we know that $\Ex{\size{w}} = \sum_{i = 1}^N \kappa_i \ge 1 \ge 1/\eps^2$ (recall $\eps \in (0, 1)$).
    If $\kappa_i < 1$ for all $i$ (i.\,e. $I = \varnothing$), we have
    $\Ex{\size{w}} = \sum_{i = 1}^N \kappa_i = \kappa \sum_{i = 1}^N p_i \overset{(\star)}{\ge} \kappa \ge 1 / \eps^2$,
    where $(\star)$ is justified by the following bound.

    \textbf{Claim:} We always have $\sum_{i = 1}^N p_i \ge 1$.

    Fix an arbitrary element $A^* \in \cD$ with $F(A^*) \neq 0$. Then
    \begin{align*}
        \sum_{i = 1}^N p_i = \sum_{i = 1}^N \max_{\substack{A \in \cD\\F(A) \neq 0}} \frac{f_i(A)}{F(A)}
        \ge \sum_{i = 1}^N \frac{f_i(A^*)}{F(A^*)}
        = \frac{1}{F(A^*)} \sum_{i = 1}^N f_i(A^*)
        = \frac{1}{F(A^*)} \cdot F(A^*) = 1.
    \end{align*}
    Noting $\Pr{\size{w} \ge \frac{3}{2} \Ex{\size{w}}} \le \Pr{|\size{w} - \Ex{\size{w}}| \ge \frac{1}{2} \Ex{\size{w}}} \le 4\eps^2$
    (as the former implies the latter) completes the proof.
\end{proof}

\begin{corollary*}[Corollary~\ref{cor:core} restated]
    In the setting of~\ralg{core}, let $\widehat{p}_1, \dots, \widehat{p}_N \in \Rnn$
    satisfy $\widehat{p}_i \ge p_i$ for all $1 \le i \le N$.
    If~\ralg{core} is executed with the $\widehat{p}_i$'s instead of $p_i = \max_{\substack{A \in \cD\\F(A)\neq 0}} \frac{f_i(A)}{F(A)}$
    in line~\ref{alg:core:pi}, it returns a vector $w \in \RR^N$ such that
    \begin{enumerate}[label = (\roman*)]
        \item $\Pr{\text{\normalfont $w$ is an $\eps$-sparsifier}} \ge 1 - \delta$,
        \item $\Ex{\size{w}} = \cO\left(\frac{\log{|\cD|} + \log{\frac{1}{\delta}}}{\eps^2} \sum_{i = 1}^N \widehat{p}_i \right)$,
        \item $\Pr{\size{w} \le \frac{3}{2} \Ex{\size{w}}} \ge 1 - 4\eps^2$.
    \end{enumerate}
\end{corollary*}

\begin{proof}
    All steps in the proof of \rtheorem{core} can be mimicked.
    For part~\ref{thm:core:sparsifier}, the magic happens when estimating the quantity $a$, where
    \begin{align*}
        \frac{f_i(A)}{\kappa \cdot \max_{S \in \cD} \frac{f_i(S)}{F(S)}}
        \le \frac{f_i(A)}{\kappa \cdot \frac{f_i(A)}{F(A)}} = \frac{F(A)}{\kappa}
    \end{align*}
    is established for all $i \in \bar{I}$.
    Since $\widehat{p}_i \ge p_i$ for each $i$, we have
    \begin{align*}
        w_i f_i(A) = \frac{f_i(A)}{\kappa \cdot \widehat{p}_i}
        \le \frac{f_i(A)}{\kappa \cdot p_i},
    \end{align*}
    for each $i$ with $\kappa \widehat{p}_i \le 1$,
    hence the same bound $a \le F(A) / \kappa$ follows.
    The indices $i$ with $\kappa \widehat{p}_i < 1$ can be handled without any modifications.
    Part~\ref{thm:core:size} is exactly the same.
    Part~\ref{thm:core:size:variance} is also largely the same, except for the claim $\sum_{i = 1}^N p_i \ge 1$ that needs to be checked for the $\widehat{p}_i$'s.
    However, this is trivial, as $\sum_{i = 1}^N \widehat{p}_i \ge \sum_{i = 1}^N p_i \ge 1$.
\end{proof}

\section{(Monotone) Submodular Functions}
\label{app:submodular}

After some preliminary results on submodular functions in
Section~\ref{subsec:prelims}, we will show in Section~\ref{subsec:counterexample} that Inequality~(\ref{ineq:peaks}), restated here as
\begin{equation*}
  \sum_{i=1}^Np_i\ \leq\ Bn,
\end{equation*}
does \emph{not} hold for general submodular $f_i$'s, thus
disproving~\cite{Rafiey22:aaai}. We will then show in
Section~\ref{subsec:submodular} that
(\ref{ineq:peaks}) does hold for monotone submodular $f_i$'s. Together
    with the $\cO(\sqrt{n})$-approximation of the peak contributions using the
    ellipsoid method~\cite{Bai16:icml}, this establishes that~\rtheorem{core}\,\ref{thm:core:size} simplifies to 
    \begin{equation}
      \Ex{\size{w}} = \cO\left( B \frac{n^2 + n \log{\frac{1}{\delta}}}{\eps^2} \right)
    \end{equation}
    if all $f_i$'s are monotone submodular,
    which is 
    \begin{equation}
      \cO\left( \frac{Bn^2}{\eps^2} \right)
     \end{equation}
     for a fixed $\delta$.

\subsection{Preliminaries} \label{subsec:prelims}

For a submodular function $f: 2^E \to \RR$, the \emph{submodular polyhedron} of
$f$~\cite{submodular-functions-and-optimization-ii} is defined as
\begin{align}
    \cP(f) := \left\{ x \in \RR^E \mid x(S) \le f(S) \text{ for all } S \subseteq E \right\}
\end{align}
and the \emph{base polyhedron} of
$f$~\cite{submodular-functions-and-optimization-ii} is defined as
\begin{align}
    \cB(f) := \left\{ x \in \cP(f) \mid x(E) = f(E) \right\},
\end{align}
where $x$ is an $|E|$-dimensional vector and $x(S)$ denotes the quantity $\sum_{e \in S} x_e$.

We denote by $\EX(P)$ the set of extreme points of a polyhedron $P$.
In particular, $\EX(\cB(f))$ is the set of extreme points of the base polyhedron of a submodular function $f$.

We will need a result that characterises the extreme points of the base polyhedron of a submodular function.

\begin{theorem}[\cite{submodular-functions-and-optimization-ii}]
\label{thm:extreme:point}
Let $f: 2^E \to \RR$ be submodular. A point $y \in \cB(f)$
is an extreme point of $\cB(f)$ if and only if for a maximal chain (with
respect to inclusion)
\begin{align}
\cC: \varnothing = S_0 \subset S_1 \subset \dots \subset S_n = E
\end{align}
of $2^E$ we have $y_{e_i} = f(S_i) - f(S_{i - 1})$ for $i = 1, \dots, n$, where $\{ e_i \} = S_i \setminus S_{i - 1}$.
\end{theorem}

\begin{corollary} \label{cor:extreme:point:monotone}
Let $f: 2^E \to \RR$ be submodular and monotone. Then
$\EX(\cB(f)) \subseteq \Rnn^E$, i.\,e., any extreme point of
the base polyhedron has only non-negative coordinates.
\end{corollary}

\begin{proof}
        Suppose $y \in \EX(\cB(f))$ is an extreme point of the base polyhedron.
        By \rtheorem{extreme:point}, there is a maximal chain $\varnothing = S_0 \subset S_1 \subset \dots \subset S_n = E$
        such that $y_{e_i} = f(S_i) - f(S_{i - 1})$ for $1 \le i \le n$, where $\{ e_i \} = S_i \setminus S_{i - 1}$.
        By monotonicity,
        \begin{align*}
            y_{e_i} = f(S_i) - f(S_{i - 1}) = f(S_{i - 1} \cup \{ e_i \}) - f(S_{i - 1}) \ge f(S_{i - 1}) - f(S_{i - 1}) = 0
        \end{align*}
        for each $1 \le i \le n$. Since the chain is maximal, each element $e \in E$ occurs as some $e_i$,
        implying $y_e \ge 0$ for all $e \in E$.
    \end{proof}

    The next lemma expresses the values $f(A)$ of a submodular function in terms of a maximisation over the base polyhedron.
    \begin{lemma}
        \label{lem:extreme:points:scalar:product}
        Let $f: 2^E \to \RR$ be normalised submodular. Then
        \begin{align*}
            f(A) = \max_{y \in \EX(\cB(f))} \left\langle y, \mathbbm{1}_{A} \right\rangle
        \end{align*}
        for all $A \subseteq E$.
    \end{lemma}
    \begin{proof}
        The inequality $\max_{y \in \EX(\cB(f))} \left\langle y, \mathbbm{1}_{A} \right\rangle \le f(A)$
        is immediate since any $y \in \cB(f)$ satisfies $y(S) \le f(S)$ for all $S \subseteq E$, including $S = A$.
        Hence
        \begin{align*}
            \langle y, \mathbbm{1}_A \rangle = \sum_{e \in E} y_e \cdot \mathbbm{1}_{A}(e) = \sum_{e \in A} y_e = y(A) \le f(A).
        \end{align*}
        To see that the value $f(A)$ is actually attained, we construct a suitable extreme point using \rtheorem{extreme:point}.
        Suppose $A = \{ e_1, \dots, e_k \}$. Order the remaining elements arbitrarily, so $E \setminus A = \{ e_{k + 1}, \dots, e_n \}$.
        Now define a maximal chain
        \begin{align*}
            \cC: \varnothing = S_0 \subset S_1 \subset \dots \subset S_n = E
        \end{align*}
        by letting $S_i := \{ e_1, \dots, e_i \}$ for $0 \le i \le n$. Next, we set
        \begin{align*}
            y_{e_i} := f(S_i) - f(S_{i - 1})
        \end{align*}
        for $1 \le i \le n$, which defines a vector $y \in \RR^E$ that -- by
        construction -- happens to be an extreme point of $\cB(f)$
        by \rtheorem{extreme:point}. Moreover, it satisfies
        \begin{align*}
            y(A) = \sum_{e \in A} y_e = \sum_{i = 1}^k y_{e_i} = \sum_{i = 1}^k \left( f(S_i) - f(S_{i - 1}) \right)
            = f(S_k) - f(\varnothing) = f(S_k) = f(A)
        \end{align*}
        as $f$ is normalised.
        Thus, $y$ witnesses that $\max_{y \in \EX(\cB(f))} \left\langle y, \mathbbm{1}_{A} \right\rangle \ge f(A)$.
    \end{proof}

\subsection{Counterexample} \label{subsec:counterexample}

   We will now disprove the bound $\sum_{i = 1}^N p_i \le Bn$ claimed by
    Rafiey and Yoshida (cf. Claim 3.3 in~\cite{Rafiey22:aaai}, proof in the
    appendix) by providing a counterexample.

    The counterexample is inspired by Cohen et al.~\cite{cohen}, who use it to establish a lower bound on the size of a cut sparsifier.
    This is exactly what we need to disprove $\sum_{i = 1}^N p_i \le Bn$: The ability to isolate edges in order to have peak contributions of $1$.
    Consider the directed complete bipartite graph $G = (V, E)$ with bipartition $V = L \cup R$ where $L = \{ u_1, \dots, u_5 \}$,
    $R = \{ v_1, \dots, v_5 \}$ and $E = L \times R$ as depicted in \rfig{peak:contributions:counterexample}.
    For each edge $e = (u, v) \in E$, we define a cut function
    \begin{align*}
        f_e: 2^V \to \RR,
        \quad S \mapsto
        \begin{cases}
            1 & \text{ if } u \in S \text{ and } v \notin S,  \\
            0 & \text{ otherwise. }
        \end{cases}
    \end{align*}
    It is well-known (and not hard to see) that the $f_e$'s are submodular. 

    The sum $F := \sum_{e \in E} f_e$
    is known as the (directed) cut function of $G$, since $F(S)$ is the number of edges cut connecting a vertex in $S$ to a vertex outside $S$.
    Let $p_e = \max_{S \subseteq V} \frac{f_e(S)}{F(S)}$ denote the peak contributions of the $f_e$'s.

    \begin{figure}[htbp]
        \centering
        \begin{tikzpicture}
            \node[draw, circle] (u1) at (0, 0) {$u_1$};
            \node[draw, circle, fill = lightgray] (u2) at (0, -1.5) {$u_2$};
            \node[draw, circle] (u3) at (0, -3) {$u_3$};
            \node[draw, circle] (u4) at (0, -4.5) {$u_4$};
            \node[draw, circle] (u5) at (0, -6) {$u_5$};

            \node[draw, circle, fill = lightgray] (v1) at (4, 0) {$v_1$};
            \node[draw, circle, fill = lightgray] (v2) at (4, -1.5) {$v_2$};
            \node[draw, circle, fill = lightgray] (v3) at (4, -3) {$v_3$};
            \node[draw, circle, fill = lightgray] (v4) at (4, -4.5) {$v_4$};
            \node[draw, circle] (v5) at (4, -6) {$v_5$};

            \draw[-stealth, dotted, thick, gray] (u1) to (v1);
            \draw[-stealth, dotted, thick, gray] (u1) to (v2);
            \draw[-stealth, dotted, thick, gray] (u1) to (v3);
            \draw[-stealth, dotted, thick, gray] (u1) to (v4);
            \draw[-stealth, dotted, thick, gray] (u1) to (v5);

            \draw[-stealth, dotted, thick, gray] (u2) to (v1);
            \draw[-stealth, dotted, thick, gray] (u2) to (v2);
            \draw[-stealth, dotted, thick, gray] (u2) to (v3);
            \draw[-stealth, dotted, thick, gray] (u2) to (v4);
            \draw[-stealth, thick] (u2) to (v5);

            \draw[-stealth, dotted, thick, gray] (u3) to (v1);
            \draw[-stealth, dotted, thick, gray] (u3) to (v2);
            \draw[-stealth, dotted, thick, gray] (u3) to (v3);
            \draw[-stealth, dotted, thick, gray] (u3) to (v4);
            \draw[-stealth, dotted, thick, gray] (u3) to (v5);

            \draw[-stealth, dotted, thick, gray] (u4) to (v1);
            \draw[-stealth, dotted, thick, gray] (u4) to (v2);
            \draw[-stealth, dotted, thick, gray] (u4) to (v3);
            \draw[-stealth, dotted, thick, gray] (u4) to (v4);
            \draw[-stealth, dotted, thick, gray] (u4) to (v5);

            \draw[-stealth, dotted, thick, gray] (u5) to (v1);
            \draw[-stealth, dotted, thick, gray] (u5) to (v2);
            \draw[-stealth, dotted, thick, gray] (u5) to (v3);
            \draw[-stealth, dotted, thick, gray] (u5) to (v4);
            \draw[-stealth, dotted, thick, gray] (u5) to (v5);
        \end{tikzpicture}
        \caption{Graph $G$ where each edge can be isolated by a cut}
        \label{fig:peak:contributions:counterexample}
    \end{figure}
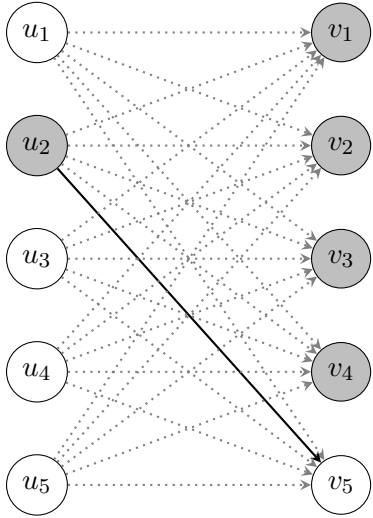

    We will now establish how this example violates $\sum_{e \in E} p_e \le Bn$.
    First,
    \begin{align*}
        B = \max_{e \in E} |\EX(\cB(f_e))| = 2
    \end{align*}
    because each base polyhedron $\cB(f_i)$ has two extreme points.
    An easy way to see this goes by \rtheorem{extreme:point}.
    For any maximal chain $\varnothing = V_0 \subset V_1 \subset \dots \subset V_n = V$,
    we have
    \begin{align*}
        f_e(V_{i + 1}) - f_e(V_i) = \begin{cases}
            1 & \text{ if } V_{i + 1} \setminus V_i = \{ u \} \text{ and } v \notin V_i, \\
            -1 & \text{ if } V_{i + 1} \setminus V_i = \{ v \} \text{ and } u \in V_i, \\
            0 & \text{ otherwise. }
        \end{cases}
    \end{align*}
    Thus, the only extreme points are $y$ and $y'$, where $y_e = 0$ for all $e \in E$,
    and $y'_u = 1$, $y'_v = -1$, $y'_e = 0$ for all $e \in E \setminus \{ u, v \}$.
    Since $G$ has $10$ vertices and the vertex set $V = L \cup R$ is the ground set, Claim 3.3 in~\cite{Rafiey22:aaai} asserts that $\sum_{e \in E} p_e \le Bn = 20$.
    By showing $p_e = 1$ for each $e \in E$, we establish
    \begin{align*}
        \sum_{e \in E} p_e = |E| = |L| \cdot |R| = 5 \cdot 5 = 25 > 20,
    \end{align*}
    contradicting the assertion.
    Fix an edge $e = (u, v)$. Letting $S := \{ u \} \cup \left( R \setminus \{ v \} \right)$, we get $f_e(S) = 1$
    and $f_{e'}(S) = 0$ for each edge $e' \neq e$. The latter is because $S$ either
    \begin{itemize}
        \item contains the right endpoint of $e'$, or
        \item does not contain the left endpoint of $e'$.
    \end{itemize}
    These are mutually exclusive. Both cases lead to $f_{e'}(S) = 0$ by definition of $f_{e'}$.
    Thus,
    \begin{align*}
        p_e = \max_{A \subseteq V} \frac{f_e(A)}{F(A)}
        \ge \frac{f_e(S)}{F(S)}
        = \frac{f_e(S)}{f_e(S) + \sum_{e' \neq e} f_{e'}(S)}
        = \frac{1}{1 + 0} = 1,
    \end{align*}
    implying $p_e = 1$, as desired.

\subsection{Monotone Submodularity} \label{subsec:submodular}

    Now we have all tools required to establish an upper bound on $\sum_{i =
    1}^N p_i$, making the size bound in \rtheorem{core}\,\ref{thm:core:size}
    simpler and more expressive.
    
    The following lemma appears as Claim 3.3 in~\cite{Rafiey22:aaai} with a very similar proof.
    However, the proof of Lemma~\ref{lem:pi:submodular} below only works if the coordinates of the extreme points of the base polyhedra are non-negative.
    We give greater detail and explanations of the individual steps and point out where exactly the non-negativity is needed,
    hence where the proof in~\cite{Rafiey22:aaai} fails. \rcorollary{extreme:point:monotone} ensures non-negativity when
    all $f_i$'s are monotone.
    For the following lemma, we specialise to $\cD = 2^E$ for some ground set $E$ of size $n$.

    \begin{lemma}
        \label{lem:pi:submodular}
        Let normalised, monotone submodular functions $f_1, \dots, f_N: 2^E \to \RR$ be given
        with peak contributions defined with respect to $F = f_1 + \dots + f_N$.
        Then $\sum_{i = 1}^N p_i \le Bn$, where $B = \max_{1 \le i \le N}
        |\EX(\cB(f_i))|$
        denotes the maximum number of extreme points of the base polyhedra of the $f_i$'s.
    \end{lemma}
    \begin{proof}
        The key ingredient is \rlemma{extreme:points:scalar:product}. It gives us
        \begin{align}
            \label{eq:pi:submodular:1}
            \sum_{i = 1}^N p_i
            = \sum_{i = 1}^N \max_{A \subseteq E} \frac{f_i(A)}{F(A)}
            = \sum_{i = 1}^N \max_{A \subseteq E} \frac
                {\max_{y \in \EX(\cB(f_i))} \langle y, \mathbbm{1}_A \rangle}
                {\sum_{j = 1}^N \max_{y \in \EX(\cB(f_j))} \langle y, \mathbbm{1}_A \rangle}
        \end{align}
        using the definition of the $p_i$'s.
        Since all $f_i$'s are monotone, we know that $y \in \Rnn^E$ for all $y
        \in \cB(f_i)$ for every $1 \le i \le N$,
        so all terms $\langle y, \mathbbm{1}_A \rangle$ involved in the above equation are non-negative.
        For non-negative numbers $x_1, \dots, x_n \in \Rnn$, we always have
        \begin{align}
            \label{eq:pi:submodular:aux}
            \frac{x_1 + \dots + x_n}{n} \overset{(1)}{\le} \max_{1 \le i \le n} x_i \overset{(2)}{\le} x_1 + \dots + x_n,
        \end{align}
        i.\,e., the maximum lies between the average and the sum.
        This is a trick used in the proof by Rafiey and Yoshida~\cite{Rafiey22:aaai} that does not hold for general submodular functions.
        Non-negativity plays an important role in a couple of steps in the remainder of this proof.
        Applying (1) to the denominator and (2) to the numerator
        in \requation{pi:submodular:1}, we obtain
        \begin{align}
            \label{eq:pi:submodular:2}
                \frac
                {\max_{y \in \EX(\cB(f_i))} \langle y, \mathbbm{1}_A \rangle}
                {\sum_{j = 1}^N \max_{y \in \EX(\cB(f_j))} \langle y, \mathbbm{1}_A \rangle}
            &\le \frac
                {\sum_{y \in \EX(\cB(f_i))} \langle y, \mathbbm{1}_A \rangle}
                {\sum_{j = 1}^N \frac{1}{|\EX(\cB(f_j))|} \sum_{y \in \EX(\cB(f_j))} \langle y, \mathbbm{1}_A \rangle} \\
            \label{eq:pi:submodular:3}
            &= \frac
                {\sum_{e \in A} \sum_{y \in \EX(\cB(f_i))} y_e}
                {\sum_{e \in A} \sum_{j = 1}^N \frac{1}{|\EX(\cB(f_j))|} \sum_{y \in \EX(\cB(f_j))} y_e},
        \end{align}
        where the equality step is by $\langle y, \mathbbm{1}_A \rangle = \sum_{e \in A} y_e$.
        Next, note that $\sum_{y \in \EX(\cB(f_i))} y_e \ge 0$
        and $\sum_{j = 1}^N \frac{1}{|\EX(\cB(f_j))|} \sum_{y \in \EX(\cB(f_j))} y_e \ge 0$ for all $e \in A$
        because $y_e \ge 0$ for all $y \in \cB(f_i)$ and $e \in E$, so we can employ the inequality
        \begin{align}
            \frac{x_1 + \dots + x_n}{q_1 + \dots + q_n} \le \max_{1 \le i \le n} \frac{x_i}{q_i}
        \end{align}
        that applies to all non-negative numbers $x_1, \dots, x_n, q_1, \dots, q_n \in \Rnn$, giving
        \begin{align}
            \frac
                {\sum_{e \in A} \sum_{y \in \EX(\cB(f_i))} y_e}
                {\sum_{e \in A} \sum_{j = 1}^N \frac{1}{|\EX(\cB(f_j))|} \sum_{y \in \EX(\cB(f_j))} y_e}
            \le \max_{e \in A} \frac
                {\sum_{y \in \EX(\cB(f_i))} y_e}
                {\sum_{j = 1}^N \frac{1}{|\EX(\cB(f_j))|} \sum_{y \in \EX(\cB(f_j))} y_e}.
        \end{align}
        Combining this with \requation{pi:submodular:2} and \requation{pi:submodular:3}, we can extend \requation{pi:submodular:1} to
        \begin{align}
            \label{eq:pi:submodular:4}
            \sum_{i = 1}^N p_i
            &\le \sum_{i = 1}^N \max_{A \subseteq E} \max_{e \in A} \frac
                {\sum_{y \in \EX(\cB(f_i))} y_e}
                {\sum_{j = 1}^N \frac{1}{|\EX(\cB(f_j))|} \sum_{y \in \EX(\cB(f_j))} y_e} \\
            \label{eq:pi:submodular:5}
            &= \sum_{i = 1}^N \max_{e \in E} \frac
                {\sum_{y \in \EX(\cB(f_i))} y_e}
                {\sum_{j = 1}^N \frac{1}{|\EX(\cB(f_j))|} \sum_{y \in \EX(\cB(f_j))} y_e} \\
            \label{eq:pi:submodular:6}
            &\le \sum_{i = 1}^N \sum_{e \in E} \frac
                {\sum_{y \in \EX(\cB(f_i))} y_e}
                {\sum_{j = 1}^N \frac{1}{|\EX(\cB(f_j))|} \sum_{y \in \EX(\cB(f_j))} y_e},
        \end{align}
        where \requation{pi:submodular:5} follows because maximising over all $A \subseteq E$ and $e \in A$ is the same as
        maximising over all $e \in E$ here,
        since the RHS only depends on $e$. \requation{pi:submodular:6} follows from \requation{pi:submodular:aux} again --
        with all quotients being non-negative because the numerators and denominators are.
        A few algebraic steps get us to the claim now. Swapping $\sum_{i = 1}^N \dots$ with $\sum_{e \in E} \dots$
        and using $|\EX(\cB(f_j))| \le \max_{1 \le \ell \le n} |\EX(\cB(f_{\ell}))| = B$, we get
        \begin{align}
            &\sum_{i = 1}^N \sum_{e \in E} \frac
                {\sum_{y \in \EX(\cB(f_i))} y_e}
                {\sum_{j = 1}^N \frac{1}{|\EX(\cB(f_j))|} \sum_{y \in \EX(\cB(f_j))} y_e} \\
            &= \sum_{e \in E} \sum_{i = 1}^N \frac
                {\sum_{y \in \EX(\cB(f_i))} y_e}
                {\sum_{j = 1}^N \frac{1}{|\EX(\cB(f_j))|} \sum_{y \in \EX(\cB(f_j))} y_e} \\
            &\le \sum_{e \in E} \sum_{i = 1}^N \frac
                {\sum_{y \in \EX(\cB(f_i))} y_e}
                {\frac{1}{B}\sum_{j = 1}^N \sum_{y \in \EX(\cB(f_j))} y_e} \\
            &= \sum_{e \in E} B = Bn,
        \end{align}
        leading to the claimed $\sum_{i = 1}^N p_i \le Bn$ via \requation{pi:submodular:4},
        \requation{pi:submodular:5} and \requation{pi:submodular:6}.
    \end{proof}

    What makes this lemma interesting is that the bound $\sum_{i = 1}^N p_i \le Bn$ does not depend on $N$.
    In cases where sparsification is interesting, we usually have $N \gg n$.

    \begin{remark}
        \label{remark:monotone}
        If all $f_i$'s are monotone, we can at least say the following, which Rafiey and Yoshida~\cite{Rafiey22:aaai} already observed.
        There is an algorithm based on the ellipsoid method that,
        given two monotone submodular functions $f, g: 2^E \to \Rnn$, approximates $\max_{A \subseteq E} \frac{f(A)}{g(A)}$
        up to a factor of $\cO(\sqrt{n} \log{n})$ in polynomial time (see Theorem 3.4 in~\cite{Rafiey22:aaai} and~\cite{Bai16:icml}).
        It can be used to obtain approximations $\widehat{p}_1, \dots, \widehat{p}_N$ such that $p_i \le \widehat{p}_i \le \cO(\sqrt{n} \log{n}) p_i$
        for $1 \le i \le N$. By \rcorollary{core}, we can execute \ralg{core} with the $\widehat{p}_i$'s and obtain an $\eps$-sparsifier of
        expected size $\cO\left( \frac{n + \log{\frac{1}{\delta}}}{\eps^2} \sum_{i = 1}^N \widehat{p}_i \right)$ in polynomial time.
        Since the $f_i$'s are monotone, \rlemma{pi:submodular} applies and the bound becomes
        $\cO\left( \frac{B n^{2.5} \log{n}}{\eps^2} \right)$
        for fixed $\delta$.
    \end{remark}

\section{Monotone Submodular Functions of Bounded Arity}
\label{app:bounded-arity}

The motivation of submodular sparsification partly arises from the
sparsification of graph cuts, which have been extensively studied not only for
graphs but also for
hypergraphs~\cite{Soma19:soda,Kogan15:itcs,Chen20:focs}.
An improvement in both size and construction time of the $\eps$-sparsifier
is possible if the submodular function $F$ in question can be decomposed as $F = f_1 + \dots + f_N$ into $f_i$'s of \emph{bounded arity}.
This covers a broad range of submodular functions such as cut functions.
As a simple application of our core algorithm, we will show that an $\eps$-sparsifier of expected size $\cO(n^2
/ \eps^2)$ can be found (for a fixed $\delta$) if all constituent functions
$f_i$ of $F$ are monotone (needed in Lemma~\ref{lem:pi:submodular}) and submodular (needed in Lemma~\ref{lem:bounded:arity:reduction} and Lemma~\ref{lem:bounded:arity:extreme:points}).

In detail, consider a set function $f: 2^E \to \RR$.
Instead of writing $f(S)$ for $S \subseteq E$ to describe the values of $f$,
we may encode $S$ as a Boolean vector of $|E|$ coordinates.
Fix an arbitrary order of elements in $E$, say $E = \{ e_1, \dots, e_n \}$ and set the $i^\text{th}$ coordinate of the vector corresponding to $S$ to $1$
if and only if $e_i \in S$. This standard identification allows us to write $f(x_1, \dots, x_n)$ for $(x_1, \dots, x_n) \in \{ 0, 1 \}^n$ to
describe the values of $f$.
From this perspective, it is intuitive to say that $f$ has \emph{arity $a$} if its values only depend on $a$ arguments.
It is worth noting that this representation aligns with how one would represent $f$
in the context of constraint satisfaction problems, as part of a valued constraint language~\cite{KTZ15:sicomp}.
Note that $f$ having arity $a$ means there exists a set $C \subseteq E$ of size $|C| = a$
such that $f(S) = f(S \cap C)$ for all $S \subseteq E$.
We call $C$ the \emph{effective support} of $f$. 

We start from a decomposition $F = f_1 + \dots + f_N$
of a submodular function $F$ into submodular functions $f_1, \dots, f_N$ of arity $\le a$ for some constant $a$.
Let $C_1, \dots, C_N$ denote the effective supports of $f_1, \dots, f_N$, respectively.
In order to construct a sparsifier, we want to apply \ralg{core}.
To do so efficiently, we need to be able to compute the peak contributions, i.\,e., maximise the quotients $\frac{f_i(A)}{F(A)}$ over all $A \subseteq E$
for each $1 \le i \le N$. Since the numerators only depend on how $A$ intersects the $C_i$'s -- which are of constant size --
we might brute-force this intersection and minimise the denominator with respect
to an extra constraint.

\begin{remark}
    To make this idea work, it is vital for us to \emph{know} the $C_i$'s.
    There might be scenarios where all $f_i$'s are known to have bounded arity while the elements they are supported on are not computationally available.
    However, if $f$ is monotone and non-negative, the effective support can be found by testing each element $e \in E$
    for $f(\{ e \}) > f(\varnothing)$. If the test fails, we know that $e$ does not influence the values of $f$ and hence does not belong to its effective support.
\end{remark}

\begin{algorithm}[htbp]
    \begin{algorithmic}[1]
        \Require Functions $f_1, \dots, f_N$ with effective supports $C_1, \dots, C_N$ of size $\le a$
        \Ensure $p_i = \max_{A \subseteq E} \frac{f_i(A)}{F(A)}$ for each $1 \le i \le N$
        \For{$i = 1, \dots, N$}
            \State $p_i \gets 0$
            \ForAll{$H \subseteq C_i$}
                \State Compute a minimiser $A^*$ of $\widehat{F}: 2^{E \setminus C_i} \to \RR, A \mapsto F(A \cup H)$
                \State $p_i \gets \max\left\{ p_i, \frac{f_i(H)}{\widehat{F}(A^*)} \right\}$
            \EndFor
        \EndFor
        \State \Return $p_1, \dots, p_N$
    \end{algorithmic}
    \caption{Computing the peak contributions for $f_i$'s of bounded arity}
    \label{alg:bounded:arity}
\end{algorithm}

Refer to \ralg{bounded:arity} for a formal description of how the peak contributions are computed.
For each $1 \le i \le N$, the quantity $p_i$ is computed as follows. We first choose a set $H \subseteq C_i$
and then restrict ourselves to sets $A \subseteq E$ with $A \cap C_i = H$.
This implies $f_i(A) = f_i(A \cap C_i) = f_i(H)$ as $C_i$ is the effective support of $f_i$.
Hence, we are left with the quotient $f_i(H) / F(A)$ that has to be maximised over all $A \subseteq E$ with $A \cap C_i = H$.
Since the numerator no longer depends on $A$, the objective becomes minimising $F(A)$ subject to $A \cap C_i = H$.
It turns out that this is just a submodular minimisation problem of an auxiliary function $\widehat{F}$.
The following statement is key to the correctness of this approach.

\begin{lemma}
    \label{lem:bounded:arity:reduction}
    Let $f, F: 2^E \to \RR$ be submodular functions and let $C \subseteq E$ denote the effective support of $f$.
    Then
    \begin{align*}
        \max_{A \subseteq E} \frac{f(A)}{F(A)} = \max_{H \subseteq C} \frac{f(H)}{\min_{A \subseteq E \setminus C} \widehat{F}(A)}
    \end{align*}
    where $\widehat{F}: 2^{E \setminus C} \to \RR, A \mapsto F(A \cup H)$.
\end{lemma}
\begin{proof}
    ``{$\le$}'' Let $A^*$ be a maximiser of the LHS.
    Set $H := A^* \cap C$ and $A := A^* \setminus C$.
    We now have
    \begin{align*}
        \frac{f(A^*)}{F(A^*)}
        = \frac{f\left(A^* \cap C\right)}{F\left(\left(A^* \setminus C\right) \cup \left(A^* \cap C\right) \right)}
        = \frac{f(H)}{F\left(\left( A^* \setminus C \right) \cup H\right)}
        = \frac{f(H)}{\widehat{F}(A)}
    \end{align*}
    with $H \subseteq C$ and $A \subseteq E \setminus C$, so this is certainly $\le$ the RHS above.

    ``{$\ge$}'' Let $(H^*, A^*)$ be a maximising pair of the RHS.
    Letting $A := A^* \cup H^*$, we get
    \begin{align*}
        \frac{f(H^*)}{\widehat{F}(A^*)}
        = \frac{f(A \cap C)}{F(A^* \cup H)}
        = \frac{f(A)}{F(A)},
    \end{align*}
    which is always $\le$ the LHS as the maximisation is over all $A$.
\end{proof}

One more observation allows us to give a more precise size bound on the sparsifier in the bounded arity case.
By \rtheorem{core}, the $\eps$-sparsifier produced by \ralg{core} is of expected size
$\cO\left(\frac{\log{|\mathcal{D}|} + \log{\frac{1}{\delta}}}{\eps^2} \sum_{i = 1}^N p_i \right)$,
where $\sum_{i = 1}^N p_i \le Bn$ by \rlemma{pi:submodular}, if the $f_i$'s are all monotone.
In general, the quantity $B$ might be large. However, if the functions all have bounded arity, we can show that $B = \cO(1)$,
allowing for an $\eps$-sparsifier of expected size $\cO(n^2 / \eps^2)$ in case
of monotone constituent functions -- assuming $\delta$ is
constant.

The following statement claimed but not shown in~\cite{Rafiey22:aaai} gives a
bound on $B$ that is independent of $n$ or $N$ but just depends on the arity.
For completeness, we give a proof.

\begin{lemma}
    \label{lem:bounded:arity:extreme:points}
    Suppose $f: 2^E \to \RR$ is submodular with arity $\le a$.
    Then the base polyhedron $\cB(f)$ has at most $2^{a^2}$ extreme points, i.\,e., $|\EX(\cB(f))| \le 2^{a^2}$.
\end{lemma}
\begin{proof}
    By \rtheorem{extreme:point},
    a point $x \in \cB(f)$ is an extreme point if and only if there is a maximal chain $\varnothing = S_0 \subset S_1 \subset \dots \subset S_n = E$
    with $x(e_i) = f(S_i) - f(S_{i - 1})$ for $1 \le i \le n$, where $\{ e_i \} = S_i \setminus S_{i - 1}$.
    Letting $C$ denote the effective support of $f$, these equations become $x(e_i) = f(S_{i} \cap C) - f(S_{i - 1} \cap C)$.
    Now $e_i \notin C$ implies
    \begin{align*}
        S_i \cap C = (S_{i - 1} \cup \{ e_i \}) \cap C = S_{i - 1} \cap C,
    \end{align*}
    hence $x(e_i) = f(S_{i} \cap C) - f(S_{i - 1} \cap C) = f(S_{i - 1} \cap C) - f(S_{i - 1} \cap C) = 0$.
    Therefore, only the coordinates corresponding to elements $e \in C$ can be non-zero.
    For any such coordinate $e \in C$, we have
    \begin{align*}
        x(e) = f(S \cup \{ e \}) - f(S) = f\left( \left(S \cup \{ e \} \right) \cap C\right) - f(S \cap C)
    \end{align*}
    where $S$ is the (unique) set in the maximal chain $e$ is added to.
    Now observe that the RHS can attain at most $2^a$ possible values: One for each subset $S \subseteq C$.
    This is because the value is completely determined by how $S$ intersects $C$.

    Putting it all together, we have at most $a$ non-zero coordinates with at most $2^a$ distinct values for each of them.
    Counting all combinations, we retrieve an upper bound of $\left( 2^a \right)^a = 2^{a^2}$ candidates for $x$.
    Thus, $\cB(f)$ has at most $2^{a^2}$ extreme points.
\end{proof}

\end{document}